\documentclass[a4paper, 11pt]{article}
\usepackage{srcltx}
\usepackage{indentfirst}
\usepackage{graphicx}
\usepackage{overpic}
\usepackage{multirow}
\usepackage[dvipdfmx]{hyperref}
\usepackage[hang]{subfigure}
\usepackage{amsmath, amssymb, amsthm}
\usepackage{color}
\usepackage[mathscr]{euscript}

\newtheorem{theorem}{Theorem}
\newtheorem{lemma}{Lemma}

\newcommand\bbR{\mathbb{R}}

\newcommand\dd{\,\mathrm{d}}
\newcommand\de{\mathrm{e}}

\newcommand\calT{\mathcal{T}}

\newcommand\pd[2]{\dfrac{\partial {#1}}{\partial {#2}}}
\newcommand\opd[2]{\dfrac{\dd {#1}}{\dd {#2}}}

\newcommand \ri{\mathrm{i}}

\numberwithin{equation}{section}
\graphicspath{{images/}}

{\theoremstyle{remark} \newtheorem{remark}{Remark}}

\title{Stationary Wigner Equation with Inflow Boundary Conditions:
  Will a Symmetric Potential Yield a Symmetric Solution?}

\author{Ruo Li\thanks{HEDPS \& CAPT, LMAM \& School of Mathematical
    Sciences, Peking University, Beijing, China, email: {\tt
      rli@math.pku.edu.cn}.}, ~~ Tiao Lu\thanks{HEDPS \& CAPT, LMAM \&
    School of Mathematical Sciences, Peking University, Beijing,
    China, email: {\tt tlu@math.pku.edu.cn}.},
    ~~ Zhangpeng Sun \thanks{School of Mathematical Sciences, 
	  Peking University, Beijing, China, 
	  email: {\tt sunzhangpeng@pku.edu.cn}.}
}

\begin{document}
\maketitle
%\input{article_abs_intro.tex}
% vim: tw=70:spell
\begin{abstract}
  Based on the well-posedness of the stationary Wigner equation with
  inflow boundary conditions given in \cite{ALZ00}, we prove without
  any additional prerequisite conditions that the solution of the
  Wigner equation with symmetric potential and inflow boundary
  conditions will be symmetric. This improve the result in
  \cite{Taj2006} which depends on the convergence of solution
  formulated in 
the Neumann series. By numerical studies, we present
  the convergence of the numerical solution to the symmetric profile
  for three different numerical schemes.
This implies that the upwind schemes can also yield a symmetric 
numerical solution, on the contrary to the argument given 
in \cite{Taj2006}.
  
  \vspace*{4mm}
\end{abstract}

\section{Introduction} \label{sec:intro} 
The stationary dimensionless Wigner equation can be written as
\cite{Wigner1932}
\begin{equation}\label{eq:Wigner}
v \pd{f(x,v)}{x}  + \int \dd v' V_{w}(x,v-v') f ( x, v') = 0 ,  
\end{equation}
where the Wigner potential $V_{w}(x,v)$ is related to the potential
$V(x)$ through 
\begin{equation}\label{eq:Vw}
V_{w}(x,v) = \frac{\ri}{2\pi} \int \dd y ~\de^{-\ri v y } \left[
V(x+y/2)-V(x-y/2) \right].
\end{equation}
We are considering the inflow boundary conditions proposed in
\cite{Frensley1987} and analyzed in \cite{ALZ00,Taj2006}, which
specifies the inflow electron distribution function $f(-l/2,v)$, $v>0$
at the left contact ($x=-l/2$) and $f(l/2,v)$, $v < 0$ at the right
contact ($x=l/2$). In \cite{Taj2006}, the Wigner equation with a
symmetric potential ($V(-x)=V(x)$) is considered. It was declared in
\cite{Taj2006} that \eqref{eq:Wigner} with a symmetric potential gives
always a Wigner function symmetric in the spatial coordinate:
$f(x,v)=f(-x,v)$, no matter what profile of the injected carrier
distribution is. Actually, it is not true for all the symmetric potential
functions. For example, when $V(x) = 1-x^2/2$, the Wigner equation
will reduce to its classical counterpart, the Boltzmann equation (the
	Liouville equation) 
\begin{equation}\label{eq:Boltzmann}
  v \pd{f(x,v)}{x} + x \pd{f(x,v)}{v} = 0. 
\end{equation}
One may figure out the solution of \eqref{eq:Boltzmann} by examining
rolling balls to a hill with a shape $V(x) = 1-x^2/2$. When one roll a
ball with the initial kinetic energy less than the height of $V(x)$
(which is $1$ at $x=0$), it is impossible to find the ball on the
right hand side. This implies that a symmetric potential can not
assure a symmetric distribution. Let us put forward a question: 

\begin{center}
\begin{minipage}{0.8\textwidth}
  {\it For which class of symmetric potential the equation
    \eqref{eq:Wigner} always has a symmetric solution for any inflow
    boundary conditions?}
\end{minipage}
\end{center}

In this paper, we answer this question partly by proving that for a
symmetric and periodic potential with a period $l$, the Wigner
equation \eqref{eq:Wigner} with inflow boundary conditions has one and
only one symmetric solution. The proof hereafter is based on the
elegant approach of the well-posedness of the stationary Wigner
equation with inflow boundary conditions in \cite{ALZ00}. The proof in
\cite{ALZ00} is given only for the discrete velocity version of the
Wigner equation providing that $0$ is excluded from the discrete
velocity points adopted. What under our consideration is the
continuous version of the Wigner equation \eqref{eq:Wigner} with the
periodic condition of the potential function. By the periodicity of
$V(x)$, we first simplify the Wigner equation to a form equivalent to
its discrete velocity version. Then we are able to make use of the
well-posedness theorem in \cite{ALZ00} to prove the symmetric
property.

In \cite{Taj2006}, a center finite-difference method was proposed to
provide a symmetric solution. It was declared therein that the
numerical solution will give an asymmetric solution in case of the
first-order upwind finite difference scheme used. This indicates that
the first-order upwind finite difference scheme will not converge to
the exact solution at all, which is predicted theoretically to be
symmetric in $x$. It argued that the strange numerical behavior is due
to that the center scheme is more physical than the upwind scheme. On
doubt of this point of view, we revisit the numerical example in
\cite{Taj2006} using three different numerical schemes, including the
two schemes used in \cite{Taj2006} and a second-order upwind finite
difference schemes. Our numerical results demonstrate that the
first-order finite difference method and the second-order upwind finite
difference scheme can also give symmetric solutions as long as the
gird size is small enough. Moreover, the symmetry of the numerical
solution can be quantitatively bounded by the accuracy of the
numerical solution. Thus it is found out that whether the numerical
solution is symmetric is not only related to numerical scheme, 
but also related to the numerical accuracy.

The remain part of this paper is arranged as below: in Section 2, we
prove the symmetry of the solution of \eqref{eq:Wigner} with symmetric
potential and in Section 3, the numerical study of the example in
\cite{Taj2006} is presented.

%%% Local Variables: 
%%% mode: latex
%%% TeX-master: "article"
%%% End: 
%\input{article_proof.tex}

% vim: tw=70:spell

\section{Symmetry of Solution of \eqref{eq:Wigner} with Symmetric Potential}
In general, the well-posedness of the boundary value problem (BVP) for the stationary Wigner equation  
\begin{equation} \label{eq:BVP}
v \pd{f(x,v)}{x} + \int \dd v' V_{w}(x,v-v')f(x,v') = 0,
  x\in(-l/2,l/2), v \in \bbR,     
\end{equation}
with the inflow boundary conditions  
\begin{equation} \label{eq:BVPBC}
f(-l/2,v)=f_b(v), \text{ for } v>0; \quad  f(+l/2,v)=f_b(v), \text{
  for } v<0, 
\end{equation}
is an open problem \cite{ALZ00}. 

At first, we expand the potential $V(x)$ into a Fourier series. The
Fourier series is uniformly converged to $V(x)$ under mild conditions,
e.g., if $V(x)$ is periodic, continuous, and its derivative $V'(x)$ is
piecewise continuous. Particularly, if $V(x)$ is symmetric with
respect to $y$-axis, i.e. it is an even function, we can expand it to
cosine series.  In this paper, we consider a special
case in which the potential function $V(x)$ defining $V_w$ through
\eqref{eq:Vw} is a periodic ($V(x+l)=V(x)$), even function with an 
absolutely convergent Fourier series, i.e., 
\begin{equation} \label{eq:VxExpansion}
  V(x)=a_0+\sum_{n=1}^{\infty}a_n\, \cos(2n\kappa x), 
\end{equation}
where $\kappa = \frac{\pi}{l}$ and $\sum_{n=0}^{\infty}|a_n|$ is
finite. Several sufficient conditions for $V(x)$ to have an absolutely
convergent Fourier series are given in \cite{Katznelson1976}, e.g., if
$V(x)$ is absolutely continuous in $[-l/2,l/2]$ and $V'(x) \in
L^2[-l/2,l/2]$, then $V(x)$ has an absolutely convergent Fourier series.
We will prove that the boundary value problem (BVP)
\eqref{eq:BVP}, \eqref{eq:BVPBC} is well-posed, and its solution is
symmetric, i.e., $f(x,v)=f(-x,v)$, $v \neq 0$, no matter what profile
of the injected carrier distribution is, providing that $V(x)$ has an
expansion \eqref{eq:VxExpansion}.  For $V(x)$ with an expansion in
\eqref{eq:VxExpansion}, a direct calculation of \eqref{eq:Vw} yields
\begin{equation} \label{eq:VwPeriodic} 
V_w(x,v) = \sum_{n=1}^{\infty} a_n \sin(2n\kappa x)
  \left(\delta(v+n\kappa)-\delta(v-n\kappa)\right).
\end{equation}
Thus, plugging \eqref{eq:VwPeriodic} into the Wigner equation
\eqref{eq:BVP}, we reformulate the stationary dimensionless Wigner
equation as  
\begin{equation} \label{eq:WignerDiff}
v\frac{\partial f(x,v)}{\partial x}+
\sum_{n=1}^{\infty} a_n \sin(2n\kappa x)
\left(f(x,v+n\kappa)-f(x,v-n\kappa)\right)=0, 
\end{equation}
where $x \in(-l/2,l/2)$ and $v \in \bbR$. 

Observing \eqref{eq:WignerDiff}, one can see that $v$ can be viewed as a
parameter, and \eqref{eq:WignerDiff} can be regarded as a set of
ordinary differential equations, and $f(x,v)$ only couples with  
$f(v+n\kappa), n \in \mathbb{Z}$.
The BVP \eqref{eq:WignerDiff} with
inflow boundary conditions \eqref{eq:BVPBC} can
be decoupled into independent ordinary differential systems indexed by
$s \in (0,\kappa)$
\begin{equation} \label{eq:BVPODE}
v_i^s\opd{f(x,v_i^s)}{ x}+
\sum_{n=1}^{\infty} a_n \sin(2n\kappa x)
\left(f(x,v_{i+n}^s)-f(x,v_{i-n}^s)\right)=0, \quad x\in(-l/2,l/2),  i \in
\mathbb{Z}, 
\end{equation}
under the inflow boundary conditions 
\begin{equation} \label{eq:BC1}
f(-l/2,v_i^s)  = f_b(v_i), \text{ for } i \geqslant  0; \quad
f(l/2,v_i^s)  = f_b(v_i),  \text{ for } i < 0, \ 
\end{equation}
where $v_i^s = i \kappa + s$. Notice that we have to neglect the case 
$s = 0$ until now.
\begin{remark}
If $s=0$, \eqref{eq:BVPODE} becomes an algebraic-differential system
and its property is quite different, and it will bring 
difficulty to the theoretical analysis \cite{ALZ00}. And
to the authors' best knowledge, $0$ is also excluded from the sampling 
velocity set in all numerical simulation papers  e.g., 
\cite{Frensley1987,Jensen1990}. 
\end{remark}

% Restate the theorem in Antorn's paper in 2000. 
% Our equation satisfy the conditions of the theorem.
% Conclusion : the equation is well-posedness.

Let $f^s_i(x)$ denote $f(x,v_i^s)$ and the vector 
$\mathbf{f}(x)=\{f^s_i(x),i\in \mathbb{Z}\}$ denote the discrete
velocity Wigner function on the discrete velocity set $\mathbf{v}^s 
= \{ v_i^s := i\kappa +s  , i \in \mathbb{Z} \}$. The discrete velocity
$v_i^s \in \mathbb{R}$ are strictly increasing, i.e.,
$v_i^s<v^s_{i+1}$. 
Considering the singularity of the equation when $v=0$, 
we have excluded $0$ from $\mathbf{v}^s$ by setting $s\neq 0$.
Henceforth, we omit the superscript $^s$
of $f^s_i(x)$, $\mathbf{f}^s$, $\mathbf{v}^s$ and $v_i^s$ when no
confusion happens.

Then we rewrite the stationary Wigner equation \eqref{eq:BVPODE} to be 
its discrete counterpart as
\begin{equation}\label{eq:BVPMatrix}
  \boldsymbol{T}\mathbf{f}_x-\boldsymbol{A}(x)\mathbf{f}=0,\,\,-l/2<x<l/2,
\end{equation}
subject to the inflow boundary conditions \eqref{eq:BC1} rewritten into  
\begin{equation}\label{eq:BC}
  f_i(-l/2)=f_b(v_i), \text{ for } i\geqslant 0, \quad
  f_i(l/2)=f_b(v_i), \text{ for } i<0,
\end{equation}
with a given sequence $\mathbf{f}_b=\{f_{b}(v_i),i\in \mathbb{Z}\}$. Here,
\begin{equation*}
  \boldsymbol{T}=\text{diag}\{\cdots, v_{-2}, v_{-1}, v_0, v_{1}, v_{2}
  \cdots  \}
\end{equation*}
and
\begin{equation*}
  \boldsymbol{A}(x)=\left(
  \begin{array}{ccccccc}
	&\ddots & & & & &\\
\cdots &a_1\sin(2\kappa x)	  &0 & -a_1\sin(2\kappa x)& -a_2\sin(4\kappa x)
	   &-a_3\sin(6\kappa x)  &\cdots\\
   \cdots &a_2\sin(4\kappa x) & a_1\sin(2\kappa x)& 0 & -a_1\sin(2\kappa x)&
					 -a_2\sin(4\kappa x)&\cdots\\
				 \cdots &a_3\sin(6\kappa x) & a_2\sin(2\kappa x) 
						& a_1\sin(2\kappa x)&0 & -a_1\sin(2\kappa x)&\cdots\\
		 & & & & &\ddots &
  \end{array}
  \right)
\end{equation*}
which is a skew-symmetric matrix.

Let the linear space $H_w = l^2(\mathbb{Z}; w_i)$ equipped with the canonical
weighted $l^2$ norm
\[ \| \mathbf{f} \|_{H_w} = \left(\sum_{i \in \mathbb{Z}} w_i | f_i |^2
\right)^{1/2},\] where $w_i = |w(v_i)|$ and $w(v)$ be a weight
function. Particularly, if $w(v) = v$, the norm of $H_v$ is
\[ \| \mathbf{f} \|_{H_v} = \left(\sum_{i \in \mathbb{Z}} |v_i| | f_i |^2
\right)^{1/2}, \] and if $w(v) = 1$, the norm of $H_1$ is
\[ \| \mathbf{f} \|_{H_1} = \left(\sum_{i \in \mathbb{Z}} | f_i |^2
\right)^{1/2}. \] We let $H = H_1$ and $B(H)$ to be the bounded linear
operator on $H$. 
We have the following
\begin{lemma}
  If the Fourier coefficients of $V(x)$, $\{a_n\}_{n=0}^{\infty} 
  \in l^1$, then $\boldsymbol{A}(x)\in B(H)$.
\end{lemma}
\begin{proof}
Observing
  \[
( \boldsymbol{A}(x) \mathbf{f}(x))_k
   =  \sum_{i=-\infty}^{\infty}
       a_{|k-i|} \sin (2(k-i)\kappa x) f_{i}(x),   
  \]
one can find that $\boldsymbol{A}(x)\mathbf{f}(x)$ 
is the discrete convolution
i.e.,
\[
 \boldsymbol{A}(x) \mathbf{f}(x) = \mathbf{v}_d(x) \ast
\mathbf{f}(x),
\]
where $\mathbf{V}_{d}(x) = \{ a_{|i|} \sin (2i\kappa x), i \in
\mathbb{Z} \}$. 
We can apply the Young's inequality to the discrete convolution  
\[
\boldsymbol{A}(x)\mathbf{f}(x) = \mathbf{V}_d(x) \ast \mathbf{f} (x) 
\]
to have
\[
\| \boldsymbol{A}(x)\mathbf{f}(x)\|_2
 =  \| \mathbf{V}_d(x) \ast \mathbf{f} (x) \| 
   _{2}  \leqslant \| \mathbf{V}_d(x) \|_1 \|\mathbf{f}(x)\|_2. 
\]
On the other hand, we have 
\[
\| \mathbf{V}_d \|_1 = \sum_{i=-\infty}^{\infty}
| a_{|i|}\sin (2i\kappa x)| \leqslant 2 \sum_{n=1}^{\infty}
|a_n| \leqslant 2 \| \{a_n\}_{n=0}^{\infty}\|_{l^1} < \infty. 
\] 
Thus by noting that the $H$-norm is the same as the $l^2$-norm, we have 
\[
\| \mathbf{A}(x) \mathbf{f}(x) \|_{H} \leqslant 2 \|\{a_{n}\}_{0}^{\infty} \|
	\|\mathbf{f}\|_{H},
\]
which gives the conclusion $\mathbf{A}(x) \in B(H)$ and 
$\| \mathbf{A}(x) \| \leq 2 \|\{a_{n}\}_{0}^{\infty} \|$. 
\end{proof}

The proof of the well-posedness of the discrete velocity problem has 
been given in \cite{ALZ00}, and here we present the conclusion below.
\begin{lemma}(Theorem 3.3 in \cite{ALZ00})
  Assume $\mathbf{f}_b=(\mathbf{f}_b^+,\mathbf{f}_b^-)\in H_v$,
  and let $\boldsymbol{A}(x)\in B(H)$ be skew-symmetric for all $x\in
  [-l/2,l/2]$. Then one has:
\item(a) If $\boldsymbol{A}\in L^1((-l/2,l/2),B(H))$, the BVP
  \eqref{eq:BVPMatrix} and \eqref{eq:BC} has a unique mild solution
  $\mathbf{f}(x)\in W^{1,1}( (-l/2,l/2), H_v)$.  Also,
  $\boldsymbol{T}\mathbf{f}_x\in L^1\left((-l/2,l/2),H \right)$;
\item(b) If $\boldsymbol{A}(x)$ is strongly continuous in x on
  $[-l/2,l/2]$ and uniformly bounded in the norm of $B(H)$ on
  $[-l/2,l/2]$, then the solution from (a) is classical, i.e.,
  $\mathbf{f}\in C^1([-l/2,l/2], H_v)$. Also,
  $\boldsymbol{T}\mathbf{f}_x\in C([-l/2,l/2],H)$.
\end{lemma}

Now we are ready to have the following theorem.
\begin{theorem}
Assume $V(x)$ is a periodic, even function with an absolutely
convergent Fourier series, i.e., $\{ a_n \}_{n=0}^\infty \in l^1$.
 For $\forall s \in (0,\kappa)$, let
  $\mathbf{f}_b = (\mathbf{f}_b^+, \mathbf{f}_b^-) \in H_v$ defined by
  \eqref{eq:BC} on $\mathbf{v}^s$. Then the BVP \eqref{eq:BVPMatrix},
  \eqref{eq:BC} has a unique solution $\mathbf{f}(x)$, 
 and for the discrete velocity Wigner function on the 
 discrete velocity set $\mathbf{v}^s$, 
$\mathbf{f}(x)$ is a mild solution in $W^{1,1}( (-l/2,l/2),H_v)$.
\end{theorem}
\begin{proof}
  It is clear that $\forall s \in (0,\kappa)$, the corresponding
  discrete velocity Wigner function $\mathbf{f}(x)$ satisfies the BVP
  \eqref{eq:BVPMatrix} and \eqref{eq:BC}. By Lemma 1, $A(x) \in B(H)$ and by
  the assumption, $\mathbf{f}_b = (\mathbf{f}_b^+, \mathbf{f}_b^-) \in
  H_v$, thus the requirement of Lemma 2 is fulfilled. Applying Lemma 2,
  we have that $\mathbf{f}(x) \in W^{1,1}( (-l/2,l/2),H_v)$. This ends
  the proof.
\end{proof}

% Based on the well-posedness, we can prove that the distribution function is
% symmetric.
Based on the well-posedness of the above Wigner equation, 
the solution of the BVP  \eqref{eq:BVPMatrix} and \eqref{eq:BC} 
satisfies the following initial value problem (IVP)
\begin{equation}\label{eq:IVP}
\opd{\mathbf{f}(x)}{x} =  T^{-1} \boldsymbol{A}(x) \mathbf{f}(x),
\quad x \in (-l/2,l/2), 
\end{equation}
with $\mathbf{f}(x)$ at $x=x_1$ is a given vector.  
One can define a propagator $\calT_{[x_1,x_2]}$ via the solution of 
the IVP \eqref{eq:IVP} \cite{Pazy1992}, 
 i.e.,
\[
\mathbf{f}(x_2) = \calT_{[x_1,x_2]} \mathbf{f}(x_1). 
\]
Clearly, the operator $\calT_{[x_1, x_2]}$ is
invertible and  \[ \calT_{[x_1, x_2]}^{-1} =
\calT_{[x_2, x_1]}. \] 
More properties of  $\calT_{[x_1,x_2]}$ can be found in 
\cite{Pazy1992}.

Moreover, if the potential is symmetric, i.e., $V(x)=V(-x)$, we have 
the following lemma.
\begin{lemma}
If $V(x)$ is a periodic, even function with an absolutely
convergent Fourier series, i.e., $\{ a_n \}_{n=0}^\infty \in l^1$,
then $\calT_{[0, x]} = \calT_{[0,-x]}$, $\forall x \in
  (-l/2, l/2)$.
\end{lemma} 
\begin{proof}
The IVP \eqref{eq:IVP} can be recast into an integral equation 
  \begin{equation}\label{eq:IntEq}
    \mathbf{f}(x) = \mathbf{f}(x_1) +
    \int_{x_1}^x T^{-1} A(y) \mathbf{f}(y) \dd y
	= \mathbf{f}(x_1)+K_{[x_1,x_2]}\mathbf{f}(x), \quad x \in [x_1,x_2]
  \text{ or } x \in [x_2,x_1],
  \end{equation}
  where $K_{[x_1,x_2]}$ is defined by   
  \[ K_{[x_1,x_2]} \mathbf{f}(x) = \int_{x_1}^{x} T^{-1} A(y)
  \mathbf{f}(y) \dd y, \text{ for } x \in [x_1,x_2] \text{ or } 
  x \in [x_2, x_1] \] 
  We have that 
  \[ \| K_{[x_1,x_2]} \mathbf{f}(x)
  \|_{H,\infty} \leqslant \dfrac{C |x_1 - x_2| }{\min_{i \in \mathbb{Z}} |v_i|} \|\mathbf{f}(x)
  \|_{H,\infty}, \]
  where the norm of the vector function $\mathbf{g}(x), x \in
  [x_1,x_2]$ is defined by 
  \[
  \|\mathbf{g}(x)\|_{H,\infty} = \text{sup}_{x\in [x_1,x_2] }
  \|\mathbf{g}(x)\|_{H}, 
 \] 
  $C$ is the twice of the $l^1$ norm of the Fourier coefficients 
  of $V(x)$. Noticing that $\min_{j \in \mathbb{Z}} |v_j| > 0$,
  we have that 
  \[ \| K_{[x_1,x_2]} \|_{C\left([x_1,x_2],H\right)} < 1, \] 
 if $|x_1 - x_2| < \tilde{\delta} = \min_{j \in \mathbb{Z}} |v_j|/C >
 0$. 
  By applying the Neumann series to \eqref{eq:IntEq}, we have 
  \[ \mathbf{f}(x) = (I-K_{[x_1,x_2]})^{-1} \mathbf{f}^{(0)}(x) =
  \sum_{n=0}^\infty K_{[x_1,x_2]}^n \mathbf{f}^{(0)}(x), \] 
or 
\begin{equation}
\mathbf{f}(x) = \lim_{n\rightarrow \infty} \mathbf{f}^{(n)}(x), 
\end{equation}
where 
\begin{equation}
\begin{split}
& \mathbf{f}^{(0)}(x) = \mathbf{f}(x_1), \quad x \in [x_1,x_2] \text{
  or } [x_2,x_1],  \\
& \mathbf{f}^{(n+1)}(x) =
\mathbf{f}(x_1)+K_{[x_1,x_2]}\mathbf{f}^{(n)}(x), ~~ x \in [x_1,x_2]
\text{ or } [x_2,x_1],
 ~ n= 0,1,\cdots.
\end{split}
\end{equation}
The Neumann series converges if $|x_2-x_1| \leq \delta <
\tilde{\delta}$. 
  $V(x)$ is symmetric, i.e., $V(x) = V(-x)$, so we see that 
  $\boldsymbol{A}(-x) = -\boldsymbol{A}(x)$. 
  At the same time, $\mathbf{f}^{(0)}(x), x \in [-\delta,\delta]$ is symmetric 
  (actually $\mathbf{f}^{(0)}(x) = \mathbf{f}(0), x \in
   [-\delta,\delta]$ is a constant vector 
   function of $x$), we have
  \[\begin{array}{rcl}
    K_{[0,-\delta]} \mathbf{f}^{(0)}(x) &=& 
	\displaystyle \int_0^{-x} T^{-1} A(y) \mathbf{f}^{(0)}(y) \dd y \\ [4mm]
    &=& \displaystyle \int_0^{x} T^{-1} A(-y) \mathbf{f}^{(0)}(-y) \dd (-y) \\ [4mm]
    &=& \displaystyle \int_0^{x} T^{-1} (-A(y)) \mathbf{f}^{(0)}(y) (-\dd y) \\ [4mm]
    &=& \displaystyle \int_0^x T^{-1} A(y) \mathbf{f}^{(0)}(y) \dd y 
     = K_{[0,\delta]} \mathbf{f}^{(0)}(x),
  \end{array}\]
  which imples $\mathbf{f}^{(1)}(-x) = \mathbf{f}^{(1)}(x)$.
It is easy to see that we can obtain  
$\mathbf{f}^{(n)}(-x) = \mathbf{f}^{(n)}(x), n=2,3,\cdots$, thus
we have $\mathbf{f}(x) = \mathbf{f}(-x), x \in [-\delta,\delta]$, i.e., 
  \[\begin{array}{rcl}
    \calT_{[0 , -x]} \mathbf{f}(0) &=& \mathbf{f}(-x) 
     = \displaystyle \sum_{n=0}^\infty K[0,-\delta]^n
	 \mathbf{f}^{(0)}(x) \\
    &=& \displaystyle \sum_{n=0}^\infty K[0,\delta]^n
\mathbf{f}^{(0)}(x) 
     = \mathbf{f}(x) = \calT_{[0 , x]} \mathbf{f}(0), \quad x \in
	 [0,\delta]. 
  \end{array}\]
  Thus $\calT_{[0, x]} = \calT_{[0, -x]}$ for $|x| \leqslant 
  \delta$. We have shown that $\mathbf{f}(-x) 
  = \mathbf{f}(x) $ for $x\in [-\delta, \delta]$. 
  We define $\mathbf{f}^{0}(x) = 
  \mathbf{f}(-\delta), x \in [-\delta, -2\delta]
$, $\mathbf{f}^{(0)}(x) = \mathbf{f}(\delta), x\in [\delta,2\delta]$,
 and it is easy to see $\mathbf{f}^{(0)}(-x)=\mathbf{f}^{(0)}(x),
  x \in [-\delta,-2\delta]\bigcup [\delta,2\delta]$.  
 Then using the same argument, we have for $x\in[-\delta,-2\delta]$,
  \[\begin{array}{rcl}
    K_{[-\delta,-2\delta]} \mathbf{f}^{(0)}(x)
  &=& \displaystyle \int_{-\delta}^{x} T^{-1} A(y) 
  \mathbf{f}^{0}(y) \dd y \\ [4mm]
    &=& \displaystyle 
\int_{\delta}^{-x} T^{-1} A(-y) 
  \mathbf{f}^{0}(-y) \dd (-y) \\ [4mm]
    &=& \displaystyle \int_{\delta}^{-x} T^{-1} (-A(y))
  \mathbf{f}^{(0)}(y) (-\dd y) \\ [4mm]
    &=& \displaystyle \int_{\delta}^{-x} T^{-1} A(y) \mathbf{f}^{(0)}(y)
  \dd y 
     = K_{[\delta,2\delta]} \mathbf{f}^{(0)}(-x), 
  \end{array}\]
and $K_{[-\delta,-2\delta]}\mathbf{f}^{(n)}(x) =
K_{[\delta,2\delta]}\mathbf{f}^{(n)}(-x), \ n = 2,3,\cdots, $
  thus
  \[\begin{array}{rcl}
  \mathbf{f}(x) = \calT_{[0,x]} \mathbf{f}(0) 
  = \calT_{[0,-x]} \mathbf{f}(0) =  \mathbf{f}(-x), ~ \forall x \in [-2\delta,-\delta]
  \bigcup [\delta,2\delta]. 
  \end{array}\]
  This implies that the domain valid for $\calT_{[0, x]} = \calT_{[0, 
  -x]}$ can be continuously extended. We conclude that $\calT_{[0, 
  x]} = \calT_{[0, -x]}$ for $\forall x \in (-l/2, l/2)$.
\end{proof}
By the lemma above, we arrive the following theorem.
\begin{theorem}
  If $V(x)$ is a periodic, even function with an absolutely convergent
  Fourier series, the solution of the BVP \eqref{eq:BVP},
  \eqref{eq:BVPBC} satisfies \[ f(x, v) = f(-x,
  v), \quad \forall v \in \mathbb{R}, v \neq n \kappa, n \in
  \mathbb{Z} \] for any inflow boundary conditions.
\end{theorem}
\begin{proof}
Since $v \neq n\kappa, n\in \mathbb{Z} $, there exists $s \in 
(0,\kappa)$ such that $v \in \mathbf{v}^s$. Thus $f(x,v)$ is an entry 
of $\mathbf{f}(x)$, which is Wigner function values at the discrete
velocity set $\mathbf{v}^s$.  From Lemma 3, we have that the discrete 
velocity Wigner function on $\mathbf{v}^s$ satisfies
  that \[ \mathbf{f}(x) = \calT_{[0, x]} \mathbf{f}(0) =
  \calT_{[0, -x]} \mathbf{f}(0) = \mathbf{f}(-x). \] 
  Then we have $f(x, v) = f(-x, v)$.
\end{proof}
\begin{remark}
  When $v=0$, the problem may be reduced to an ODE system in the same
  form as \eqref{eq:BVPMatrix} and \eqref{eq:BC}. The equation at
  $s=0$ in \eqref{eq:BVPODE} is formally turned into an algebraic
  constraint as
  \[
  \sum_{n=1}^{\infty} a_n \sin(2n \kappa x) \left( f_{n}(x)
    - f_{-n}(x)\right) = 0. 
  \]
  It is clear that if $f(x, 0) \equiv 0$, the well-posedness of the
  system is still valid. If $f(x, 0) \not \equiv 0$, this algebraic
  constraint above can not always be fulfiled, thus the existence of
  the solution is negative. One the other hand, the term involving the
  derivative in $x$ at $v = 0$ may be a $0 \times \infty$ form, since
  the derivative ${\partial f(x, 0)}/{\partial x}$ is not bounded by
  any regularity. This makes the algebraic constraint obtained above
  by formally dropping the first term is doubtable, which requires
  further investigation.
\end{remark}

%%% Local Variables: 
%%% mode: latex
%%% TeX-master: "article"
%%% End: 
%\input{article_numerical.tex}

% vim: tw=70:spell

\section{Numerical Study on Symmetry of Solution}
%We do the numerical experiment on the simplest case.
In order to verify the theoretical analysis in Theorem 2, we revisit
the example in \cite{Taj2006} which considers a particular
potential profile $V(x) = V_0 (1+\cos (2 \kappa x))$. The
corresponding Wigner potential $V_w$ in \eqref{eq:VwPeriodic} is
simply given by 
\begin{equation} \label{eq:VwNum}
V_w(x,v) = V_0 \sin (2\kappa x)
  \left(\delta(v+\kappa)-\delta(v-\kappa)\right).  
\end{equation} 
The boundary conditions are extremely simple, too. A mono-energetic 
carrier injects only from the left boundary, i.e., we set 
$f_b(v_i)$ in \eqref{eq:BC} to be   
\begin{equation} \label{eq:fb}
f_b(v_i) = \begin{cases} 1, & \text{if } i = 0, \\ 0, & \text{else},     
\end{cases}
\end{equation} 
where $v_i = (i+1/2) \kappa$ which means $s=\kappa/2$ (the index 
of the differential system \eqref{eq:BVPODE}).  

In this section, we will solve the ordinary differential system 
\eqref{eq:BVPMatrix}-\eqref{eq:BC} with $V_w$ in \eqref{eq:VwNum},
which reduces to   
\begin{equation} \label{eq:WignerNum}
v_i \opd{f_i(x)}{x} = V_0 \sin (2\kappa x) \left( f_{i+1}(x) -
	f_{i-1}(x)\right),
\quad x\in(-l/2,l/2), i \in \mathbb{Z}, 
\end{equation}
under inflow boundary conditions \eqref{eq:BC} with inflow data given 
in \eqref{eq:fb}.

Set the other parameters to be $V_0=20$, and $l=1$ which gives
$\kappa=\pi/l = \pi$.  The kinetic energy is $v_0^2/2$, which is lower
than the height of the potential energy in the middle of the device.
If using the classic mechanics (the Liouville equation or the
Boltzmann equation), one can not see the carrier in the right part of
the device. This implies the solution of the classic mechanics
equation will be asymmetric.

In the following part of the section, we
will show that the numerical solution of the Wigner equation will be
symmetric using three numerical schemes. The first scheme is the 
first-order upwind finite-difference method \cite{Frensley1987}, and 
the other two are second-order finite-difference methods. 
The second scheme is the second-order upwind finite-difference 
method used in many numerical simulation papers, e.g., 
\cite{Jensen1990}.  The authors of \cite{Taj2006}
used the first-order upwind finite-difference method and failed to 
get a symmetric solution, so they proposed a central finite-difference 
method based on physical argument, which is adopted as our third 
scheme.

\subsection{Three Finite-Difference Schemes}
We implement the finite difference methods on a uniform mesh. The
$x$-domain $[-l/2,l/2]$ is discretized with $N_x+1$ equally distanced
grid points $x_j = j\Delta x, j = 0,1,\cdots,N_x$ where $\Delta x =
l/N_x$. We have to  truncate 
$i\in \mathbb{Z}$ into a finite set $\{i: M \leqslant i \le M\}$.
$M$ depends on the inflow data and the potential strength $V_0$,
and in our current example, we set $M=40$, which is found large enough
for our numerical example. 
We denote $f_i(x)$ at $x=x_j$ by $f_{i,j}$.   
The first order upwind finite-difference scheme is obtained by
approximate $\left . \opd{f_i(x)}{x} \right|_{x_j}$ with 
\[
\left . \opd{f_i(x)}{x} \right| _{x_j} 
\approx 
\begin{cases}
\frac{ f_{i,j+1}-f_{i,j}}{\Delta x}, & \text{if } v_i < 0 , \\
\frac{ f_{i,j}-f_{i,j-1}}{\Delta x}, & \text{if } v_i > 0 , \\
\end{cases}
\]
thus yields the finite difference equations of \eqref{eq:WignerNum} as
\begin{equation}
 \left\{ \begin{array}{ll}
 v_i\dfrac{f_{i,j+1}-f_{i,j}}{\Delta x}=g_{i,j} ,
 & v_i<0, j=0,1,\cdots,N_x-1, \\
 v_i\dfrac{f_{i,j}-f_{i,j-1}}{\Delta x}=
g_{i,j} ,& v_i>0,  j=1,2,\cdots,N_x .\\
 \end{array}
  \right.
\end{equation}
where 
\begin{equation}\label{eq:gij}
 g_{i, j} = V_0 \sin(2\kappa x_j) (f_{i-1,j}-f_{i+1,j}).
\end{equation}
The second order upwind finite-difference method is obtained by approximate 
$\left . \opd{f_i(x)}{x} \right|_{x_j}$ with 
\begin{equation}
\left . \opd{f_i(x)}{x} \right|_{x_j}
\approx
\left \{
 \begin{array}{ll}
  \dfrac{-f_{i,j+2}+4 f_{i,j+1}-3 f_{i,j}}{2\Delta x} , &  v_i<0, \\ [4mm]
  \dfrac{f_{i,,j-2}-4 f_{i,j-1}+3 f_{i,j}}{2\Delta x} , &  v_i>0. 
 \end{array}
\right.
\end{equation}
The second order upwind scheme includes three nodes,
so the first order upwind scheme is used in the boundary cell
instead of the second order one. 

Both upwind schemes approximate $\opd{f_i}{x}$ at the grid point $x_j$
using the grid points on one side, while the central difference
proposed in \cite{Taj2006} approximates $\opd{f_i(x)}{x}$ at
$x_{j+1/2}= (x_j+x_{j+1})/2$ using the grid points $x_{j}$ and 
$x_{j+1}$. That is, the central scheme approximates 
$\left. \opd{f_i(x)}{x} \right|_{x_{j+1/2}}$ with 
\begin{equation}
\left . \opd{f_i(x)}{x} \right|_{x_{j+1/2}}
\approx
  \dfrac{f_{i,j+1}-f_{i,j}}{\Delta x }.   
\end{equation}
Then the right hand side of \eqref{eq:WignerNum} at $x_{j+1/2}$ is 
approximated by the average of its values at $x_{j}$ and $x_{j+1}$. 
Finally, the central scheme results in a difference equation as follows
\begin{equation}
  \left\{ \begin{array}{rl}
      v_i \dfrac{f_{i,j+1}-f_{i,j}}{\Delta x}= \dfrac{g_{i,j+1} + g_{i,j}}{2}, &
      v_i<0, j=0,1,\cdots,N_x-1, \\ [2mm]
      v_i\dfrac{f_{i,j}-f_{i,j-1}}{\Delta x}= \dfrac{g_{i,j+1} + g_{i,j}}{2}, &
      v_i>0, j=1,2,\cdots,N_x,
  \end{array}\right.
\end{equation}
where $g_{i,j}$ is given in \eqref{eq:gij}.

\subsection{Numerical Results}
We solve \eqref{eq:WignerNum} using the first order upwind 
finite-difference method on different meshes with
the grid numbers $N_x=100,200,400,800,
1600.$ The numerical results show us the Wigner 
distribution and the density function are clearly not symmetric 
in the spatial coordinate.  This is maybe the reason why the
authors of \cite{Taj2006} declares that the symmetric solution 
can not be obtained by the upwind scheme. 
But we keep refining the mesh by using $N_x=3200,\cdots,25600$, the 
numerical results tend to be symmetric as expected, which means that 
the first order upwind scheme can also give us a
symmetric numerical solution. 

We solve \eqref{eq:WignerNum} using the two second-order methods 
on the mesh with the grid number $N_x=100$.
The Wigner distribution obtained by using the central
finite-difference method with $N_x=100$ is shown on Figure 
\ref{fig-F_c1}. We can see from the figure that
\begin{enumerate}
\item The Wigner distribution function is strongly symmetric.
\item The incident particles with low energy can tunnel through the 
barrier.
\item The Wigner potential plays a role of  
a scattering mechanism and scatters carriers to higher energy state.
\item The Wigner distribution function is negative in some region,
  which is distinct from the classic distribution function.  
\end{enumerate}
\begin{figure}
 \centering
 \includegraphics[width=0.8\textwidth]{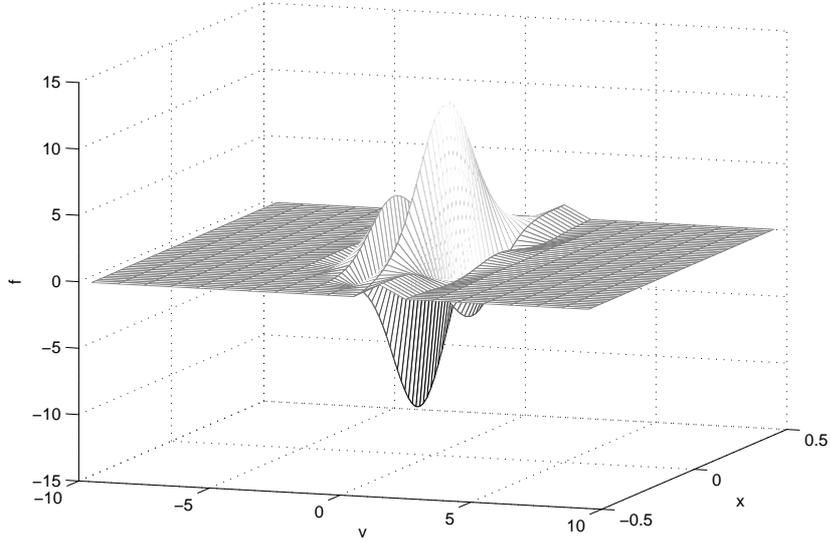}
 \caption{The distribution function obtained by using the central scheme
   on the mesh with $N_x =100$. 
   \label{fig-F_c1}}
\end{figure}

By using the three schemes, we compute the density $n(x_j)$ defined by 
\begin{equation}
	n(x_j) =\sum_{i=-\infty}^{+\infty}f_i(x_j).
\end{equation}
As shown in Figure \ref{fig-density123}, the 2nd-order methods 
give a symmetry density, and the 1st-order method also gives 
a symmetric density on a very fine mesh. 
\begin{figure}
\centering
\includegraphics[width=1.0\textwidth]{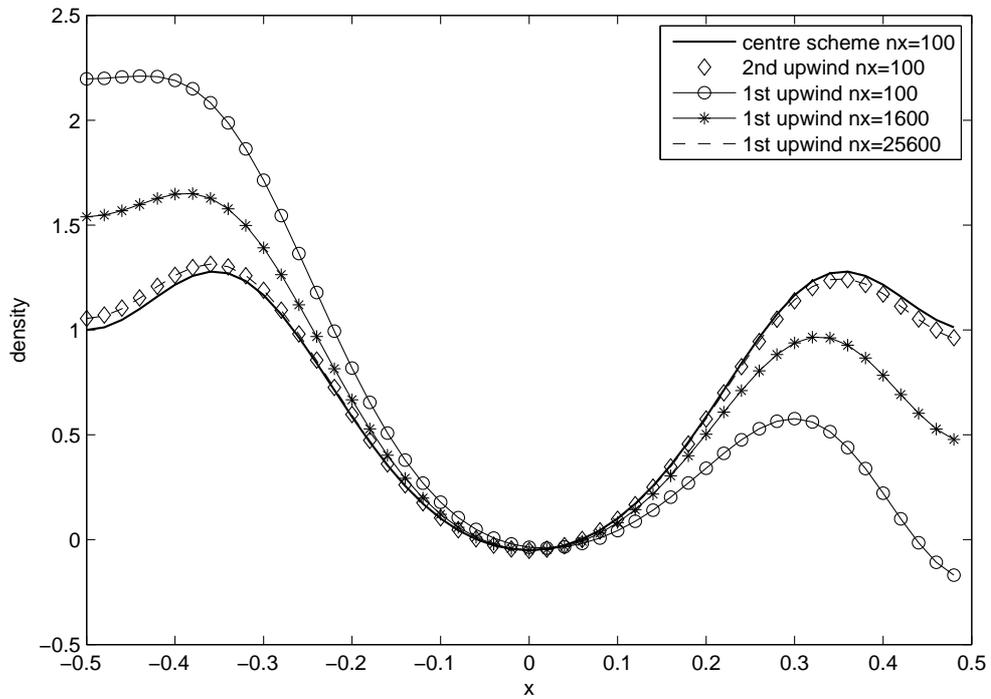}
\caption{ Density calculated by using the three schemes. \label{fig-density123}}
\end{figure}
The numerical solution of the 2nd order upwind scheme with $N_x=100$
is consistent with that of the 1st order upwind scheme with $N_x=25600$,
and the result of the central scheme is more symmetric. 
The difference between two second-order method is reflected from 
Figure \ref{fig-density123}, where the density obtained by using 
the second-order upwind finite-difference method with $N_x=100$ is 
almost coincident with that obtained by using the first-order upwind 
finite-difference method with $N_x=25600$ while the density obtained by
using the central scheme is more symmetric.

Figure \ref{fig-density123} gives us an intuitive understanding of the
symmetry of the solution. Next, we will define a symmetry error to 
compare the three schemes.  Define the symmetry error to be 
\begin{equation}
  e_{\rm sym}=\int \dd v \int \dd x |f(x,v)-f(-x,v)|.
\end{equation}
Numerically, the symmetry error can be approximated by
\begin{equation}
  \tilde{e}_{\rm sym}=\sum_{i} \sum_{j} |f_i(x_j)-f_i(-x_j)| \Delta x.
\end{equation}

The numerical symmetry errors obtained by using different schemes are
collected in Table \ref{tab:symmetric error}. It can be seen that the
numerical solution obtained by using the first order upwind scheme
becomes more and more symmetric as refining the mesh, the symmetry
error of the 2nd order upwind scheme with $N_x=100$ is about the same
with the 1st order upwind scheme with $N_x=25600$, and the solution
obtained by the central scheme is perfectly symmetric due to the
symmetry of the scheme itself. These are consistent with the results
in Figure \ref{fig-density123}.
\begin{table}[!hbp]
  \centering
  \begin{tabular}{|l|c|c|c|c|c|}
	\hline
	\hline
	$N_x$ & 100 & 400 & 1600 & 6400 & 25600 \\
	\hline
	1st upwind & 1.03 & 0.7666 & 0.4185 & 0.1502 & 0.0422\\
	\hline
	2nd upwind & 0.0462 & 7.446e-4 & 1.151e-5 & & \\
	\hline 
	central & 2.5966e-16 & & & & \\
	\hline
	\hline
  \end{tabular}
  \caption{Symmetry errors of the three schemes\label{tab:symmetric error}}
\end{table}

%%% Local Variables: 
%%% mode: latex
%%% TeX-master: "article"
%%% End: 
%\input{article_conclusion.tex}

% vim: tw=70:spell

\section{Conclusion}
For the problem whether the solution of the stationary Wigner equation
with inflow boundary conditions will be symmetric if the potential is
symmetric in \cite{Taj2006}, we give a rigorous proof based on
\cite{ALZ00} under mild assumption on the regularity of the
potential. It is concluded that a certain kind of continuous Wigner
equation with inflow boundary condition can be reduced to the
discrete-velocity case, thus is well-posed. Furthermore, we
numerically studied the example in \cite{Taj2006} and pointed out that
the numerical solution will converge to the exact solution with
symmetry, even when the numerical scheme adopted is not symmetric  
if only accuracy is enough.

\section*{Acknowledgements}
This research was supported in part by the National Basic Research
Program of China (2011CB309704) and NSFC (91230107).

%%% Local Variables: 
%%% mode: latex
%%% TeX-master: "article"
%%% End: 
%\bibliographystyle{plain}
%\bibliography{../tiao}

\end{document}